\documentclass[aps,twocolumn,floatfix,superscriptaddress,longbibliography]{revtex4}
\usepackage{graphicx,multirow}
\usepackage{color,outlines}
\usepackage[sc,osf]{mathpazo}

\usepackage{physics}
\usepackage{outlines}
\usepackage{tikz}
\usepackage{url,hyperref}
\usepackage{graphicx}
\usepackage{dcolumn}
\usepackage{bm}
\usepackage{color}
\usepackage{xcolor}
\usepackage{bbold}
\usepackage{epsfig,epic}
\usepackage{epstopdf}
\usepackage{amsthm}
\usepackage{braket}
\usepackage{float}
\newcommand{\blu}{}

\newcommand{\beq}{\begin{equation}}
\newcommand{\eeq}{\end{equation}}
\newtheorem{thm}{Theorem}
\newtheorem{corollary}{Corollary}
\newtheorem{definition}{Definition}

\newcommand{\rl}[1]{\left(#1\right)}

\newcommand{\lK}{K}
\newcommand{\lL}{L}

\newcommand{\newc}{\newcommand}
\newc{\kt}{\rangle}
\newc{\br}{\langle}
\newc{\beqa}{\begin{eqnarray}}
\newc{\eeqa}{\end{eqnarray}}
\newc{\ovl}{\overline}
\usepackage[flushleft]{threeparttable} 
\usepackage{tabularx}
\usepackage{array}

\begin{document}

\title{Absolutely maximally entangled state equivalence and the construction of infinite quantum solutions to the problem of 36 officers of Euler}
\author{Suhail Ahmad Rather \footnote{Present address: Max-Planck-Institut f\"ur Physik komplexer Systeme, N\"othnitzer Stra\ss{}e 38, 01187 Dresden, Germany} }
\thanks{Both authors contributed equally} 
\affiliation{Department of Physics, Indian Institute of Technology
Madras, Chennai~600036, India}
\affiliation{Center for Quantum Information, Communication and  Computation, Indian Institute of Technology Madras, Chennai ~600036, India}
\author{N. Ramadas}
\thanks{Both authors contributed equally}

\affiliation{Department of Physics, Indian Institute of Technology
Madras, Chennai~600036, India}
\affiliation{Center for Quantum Information, Communication and  Computation, Indian Institute of Technology Madras, Chennai ~600036, India}
\author{Vijay Kodiyalam}
\affiliation{The Institute of
	Mathematical Sciences, CIT Campus, Taramani, Chennai~600113, India} 
\affiliation{Homi Bhabha National
	Institute, Anushakti Nagar, Mumbai~400085, India}
\author{Arul
Lakshminarayan}
\affiliation{Department of Physics, Indian Institute of Technology
	Madras, Chennai~600036, India}
\affiliation{Center for Quantum Information, Communication and  Computation, Indian Institute of Technology Madras, Chennai ~600036, India}
   
\begin{abstract}
Ordering and classifying  multipartite quantum states by their entanglement content remains an open problem. One class of highly entangled states, useful in quantum information protocols, the absolutely maximally entangled (AME) ones, are specially hard to compare as all their subsystems are maximally random. 
While, it is well-known that there is no AME state of four qubits, many analytical examples and numerically generated ensembles of 
four qutrit AME states are known. However, we prove the surprising result that there is truly only {\em one} AME state of four qutrits up to local unitary equivalence.  In contrast, for larger local dimensions, the number of local unitary classes of AME states is shown to be infinite. Of special interest is the case of local dimension 6 where it was established recently that a four-party AME state does exist, providing a quantum solution to the classically impossible Euler problem of 36 officers. Based on this, an infinity of quantum solutions are constructed and we prove that these are not equivalent. The methods developed can be usefully generalized to  multipartite states of any number of particles.
\end{abstract}

\maketitle
\section{Introduction}
Quantum entanglement between two distant parties, with its counterintuitive nonclassical features, has been experimentally verified  \cite{Clauser_1972,Aspect_1982,Giustina_loopholefree_2015,Hensen_loophole_2015} via violations of Bell-CHSH inequalities \cite{Clauser_1969}. Multipartite entanglement which is at the heart of quantum information, computation and many-body physics is still poorly understood. Studying the entanglement content in them via inter-convertibility and classifying them are of fundamental importance. A putative maximally entangled class called absolutely maximally entangled (AME) states are such that there is maximum entanglement between any subset of particles and the rest \cite{Helwig_2012}. They have been related to error correcting codes \cite{Sc04}, both classical and quantum, combinatorial designs such as orthogonal Latin squares (OLS) \cite{Clarisse2005,Goyeneche2015,GRMZ_2018}, quantum parallel-teleportation and secret sharing \cite{Helwig_2012}, and holography \cite{Pastawski2015}.
{ It is therefore of considerable interest to find structure among such highly entangled multipartite states, in particular, can some AME states have more nonlocal resource than others?}

Given $N$ particles with $d$ levels each (the local dimension is $d$) there is no guarantee that an AME state, denoted AME$(N,d)$, exists. For example, AME$(4,2)$ does not exist; four qubits cannot be absolutely maximally entangled \cite{Higuchi2000}. It is known that AME$(N,2)$ exists only for $N=2,3,5,$ and $6$  \cite{Bennett_1996_QECC,Sc04, Huber_2017}. A table of known AME$(N,d)$ constructions is maintained \cite{Huber_AME_Table}, and a recent update is presumably AME$(4,6)$ \cite{SRatherAME46}. This long defied construction and provided a quantum solution to the classically impossible problem of ``36-officers of Euler". This state was dubbed the ``golden-AME" state due to the unexpected appearance of the golden ratio in it. Recent works have appeared elucidating the nature of the solution and its geometric implications ~\cite{Zyczkowski_et_al_2022,Zyczkowski_2022}.

Given that any type of entanglement cannot on the average increase under local operations and classical communication (LOCC), two states $|\psi_1\kt$ and $|\psi_2\kt$ are said to be LOCC-equivalent if they can be converted to each other under such operations \cite{Bennet_entang_concen_1996,Bennett_2000}. A finer, but easily defined, equivalence is local unitary (LU) equivalence: \begin{align}
\begin{aligned}
|\psi_1\kt \stackrel{LU}{\sim} |\psi_2\kt
\end{aligned}
\end{align} iff there exists local unitary operators $u_i$, such that $|\psi_2\kt=(u_1\otimes\cdots \otimes u_N) |\psi_1\kt$.
A coarser classification is provided by Stochastic-LOCC wherein conversion occurs with a nonzero probability of success \cite{Duur_2000,Acin_2000}. Mathematically, this replaces the unitary $u_i$ in the LU-equivalence by invertible matrices. For pure AME states such as this work addresses, SLOCC (and hence also LOCC) equivalence is identical to LU-equivalence \cite{Adam_SLOCC_2020}. 


However LU-equivalence among AME states is a long-standing problem that is notoriously hard to resolve \cite{Kraus_2010,Sauerwein_SLOCC_2018,Adam_SLOCC_2020} as all the subsystem states are maximally mixed. There have been several examples of AME$(4,3)$ in the literature, from those based on graph states and combinatorial structures \cite{Clarisse2005,Helwig_2013,Goyeneche2015,Gaeta_2015,ASA_2021} to numerically generated ensembles \cite{SAA2020, rico2020absolutely}.
Despite this, we prove the surprising conjecture \cite{SAA_2022} that there is exactly {\em one} LU-equivalence class of AME$(4,3)$ states.  We show that they are all equivalent to each other and hence equivalent to one with minimal support or rank such as: 
\begin{align}\label{eq:p9_state}
\begin{aligned}
|\Psi_{P_9}\kt=&(|1111\kt+|1222\kt +|1333\kt +|2123\kt +|2231\kt \\&+|2312\kt +|3132 \kt+|3213\kt+|3321\kt)/3. 
\end{aligned}
\end{align}
{ The minimal support or rank here refers to the minimum natural number $r$ such that the corresponding state can be represented as a superposition of $r$ orthonormal product states \cite{bruzda2023rank}. For an AME$(N,d)$ state with $N,d \geq 2$, the minimal support is found to be $d^{\lfloor N/2\rfloor}$ \cite{bernal2017existence}.}

We provide a set of invariants that when they coincide for two states implies their LU-equivalence. 
For $d>3$, but $d\neq 6$, orthogonal Latin squares (OLS) can be used to construct AME$(4,d)$ states \cite{Clarisse2005,Goyeneche2015}. We show that a continuous parametrization based on multiplication of suitable components by  phases gives invariants that can take an uncountable infinity of values and hence lead to an infinity of LU-equivalence classes. For the special case of $d=6$, there are no OLS constructions \cite{bose1960further}, but we use the recently constructed ``golden-state" \cite{SRatherAME46} as a basis for a similar construction which leads to an infinity of LU-equivalence classes in this case as well.  The methods developed can be generalized to larger number of particles and provides a new outlook into highly entangled multipartite states.

A unitary matrix $U$ of order $d^2$ can be used to define a four-party state
\beq 
\label{eq:AME_U}
\ket{\Psi_{U}}=\frac{1}{d}\sum_{i \alpha j \beta}     U^{i \alpha}_{j \beta}  |i \alpha j \beta \kt,
\eeq
where $ U^{i \alpha}_{j \beta} :=\bra{i \alpha}U \ket{j \beta}$.
The state $\ket{\Psi_{U}}$ is a vectorization of the matrix $U$ \cite{Zanardi2001,Zyczkowski2004}. If the unitary $U$ is 2-unitary (defined in next section), then the corresponding state is an AME($4,d$). A unitary operator $U$  is LU equivalent to $U'$ if there exist single-qudit gates $u_i$ and $v_i$ such that 
\begin{align}
\begin{aligned}
U'=(u_1 \otimes u_2) U (v_1 \otimes v_2).
\end{aligned}
\end{align}
The corresponding four-party states are also LU equivalent as $\ket{\Psi_{U'}}=(u_1 \otimes u_2 \otimes v_1^T \otimes v_2^T)\ket{\Psi_{U}}$, where $T$ is the usual transpose. Therefore the equivalence among AME($4,d$) states can be studied via equivalence of 2-unitary operators. In this paper, we construct and use LU invariants that are based on unitary operators rather than directly the coefficients of states. Based on four permutations of $n$ copies, these are easily computed and are in principle complete, in the sense that if all of them are equal then the operators or corresponding states are LU-equivalent \cite{VijayK}.

The fact that there is only one LU class of AME$(4,3)$ states implies that there is only {\it one} 2-unitary matrix of order $9$ denoted $P_9$, up to multiplication by local unitaries on either side, and no genuinely orthogonal quantum  Latin square \cite{MV19,GRMZ_2018} in $d=3$. It should also be noted that while generic states of four parties (even for qubits) have an infinity of LU-equivalence classes, the case of AME states forms an exceptional set.

{

\section{Preliminaries and Definitions}
In this section, we recall necessary background on classical and quantum orthogonal Latin squares, and 2-unitary operators.

\subsection{Orthogonal Latin squares}

A Latin square (LS) of order $ d $ is a $ d\times d $ array filled by numbers $[d]= \lbrace1,2,..,d \rbrace $ each appearing exactly once in each row and column. 
Two Latin squares $ K $ and $ L $ of order $ d $, with entries $ K_{ij} $ and $ L_{ij} $ in $i$-th row and $j$-th column, are called orthogonal Latin squares if the $ d^2 $ pairs $ (K_{ij},L_{ij}) ,~~i,j\in[d]$ occur exactly once.

As mentioned earlier, a pair of orthogonal Latin squares of order $d$ can be used to construct an AME$(4,d)$ state. If the orthogonal Latin squares are $K$ and $L$ with order $d$, the corresponding AME$(4,d)$ state can be constructed as follows
\begin{align}\label{eq:ame_from_ols}
\begin{aligned}
\ket{\psi_{(K,L)}} = \frac{1}{d} \sum_{i,j} \ket{ij} \ket{K_{ij}L_{ij}}.
\end{aligned}
\end{align} 
Such construction exists for all $d$, except $d=2$ and $d=6$, where there are no OLS.

{ The notion of a Latin square can be generalized by replacing discrete symbols with vectors or pure quantum states  \cite{MV16}. A quantum Latin square  of size $d$ is a $d\times d$ arrangement of $d$-dimensional vectors such that each row and column forms an orthonormal basis. The mapping of discrete symbols in a classical Latin square to computational basis vectors; $\left\lbrace i \mapsto \ket{i},j\mapsto \ket{j}:\langle i | j \rangle=\delta_{ij}, i,j=1,2,\cdots,d\right\rbrace$, results in a quantum Latin square of size $d$. For $d=2$ and $3$ all quantum Latin squares are equivalent to classical ones for some appropriate choice of bases \cite{Paczos_2021}. However, for $d\geq 4$ there exist quantum Latin squares that are not equivalent to classical Latin squares \cite{Paczos_2021}. 

Analogous to orthogonality of classical Latin squares there is a notion of orthogonality of quantum Latin squares \cite{GRMZ_2018,MV19}. Two quantum Latin squares $\mathcal{Q}_1$ and $\mathcal{Q}_2$ are said to be orthogonal if together they form an orthonormal basis in $\mathcal{H}_{d^2}$. To be precise, if $\ket{\psi_{ij}}$ and $\ket{\phi_{ij}}$ denote single-qudit states in the $i$-th row and $j$-th column of orthogonal quantum Latin squares $\mathcal{Q}_1$ and $\mathcal{Q}_2$, then the set  $\left\lbrace \ket{\psi_{ij}} \otimes \ket{\phi_{ij}}; i,j=1,2,\cdots,d\right\rbrace$ is an orthonormal basis. Thus orthogonal quantum Latin squares provide a special product basis in $\mathcal{H}_d \otimes \mathcal{H}_d$  in which both single-qudit basis form a quantum Latin square. 

A more general notion of orthogonal quantum Latin squares allows for an entangled basis in $\mathcal{H}_d \otimes \mathcal{H}_d$.  In this case a $d \times d$ array of bipartite pure states $\ket{\Psi_{ij}} \in \mathcal{H}_d^{A} \otimes \mathcal{H}_d^{B}$ form an orthogonal quantum Latin square, if they satisfy the following conditions \cite{SRatherAME46,rico2020absolutely}: 
\begin{equation*}
\begin{split}
\braket{ \Psi_{ij} | \Psi_{kl}} &= \delta_{ij} \delta_{kl}  \, ,\\
\text{Tr}_A \rl{ \sum_{k=1}^d  \ket{\Psi_{ik}} \bra{ \Psi_{jk}}}  = \delta_{ij} &\mathbb{I}_d =
 \text{Tr}_B \rl{ \sum_{k=1}^d  \ket{\Psi_{ik}} \bra{ \Psi_{jk}}} ,\\
\text{Tr}_A \rl{ \sum_{k=1}^d  \ket{\Psi_{ki}} \bra{ \Psi_{kj}}}  = \delta_{ij} &\mathbb{I}_d =
 \text{Tr}_B \rl{ \sum_{k=1}^d  \ket{\Psi_{ki}} \bra{ \Psi_{kj}}}.
\end{split}
\end{equation*}
Here $\Tr_A$ and $\Tr_B$ denotes the partial trace operations onto subsystems $B$ and $A$, respectively. It is easily seen that in the case of an unentangled basis these conditions are equivalent to the first definition. 

The above definitions are equivalent to the bipartite unitary operator $U=\sum_{i,j=1}^d \ket{i}\ket{j}\bra{\Psi_{ij}}$ remaining unitary under particular matrix rearrangements as explained below.  Such unitary operators are called 2-unitary \cite{Goyeneche2015} and form the main focus of this work.}

\subsection{2-unitary operators}
A unitary operator $U$ on $\mathbb{C}^d \otimes \mathbb{C}^d$ $\in \mathcal{U}(d^2)$ can be expanded in a product basis as 

\begin{align}
\begin{aligned}
U = \sum_{i\alpha j \beta} \bra{i \alpha} U  \ket{j \beta} \ket{i\alpha} \bra{j \beta}.
\end{aligned}
\end{align}
We recall the following matrix rearrangement operations familiar from state separability criteria \cite{chen2002matrix,peres1996separability}: 

\begin{itemize}
\item[(i)] Realignment, $R$ :
 \begin{align}
\begin{aligned}
\bra{ij}U^R\ket{\alpha \beta}=\bra{i \alpha}U \ket{j \beta}
\end{aligned}
\end{align}
\item[(ii)] Partial (or blockwise) transpose, $\Gamma$:
\begin{align}
\begin{aligned}
\bra{i\beta}U^{\Gamma}\ket{j\alpha}=\bra{i \alpha}U \ket{j \beta}.
\end{aligned}
\end{align}
\end{itemize}
Here $U^R$ and $U^\Gamma$ denote the matrices obtained after realignment and partial transpose operations, respectively. These operations allows us to define the following classes of unitary operators:
\begin{definition}
(Dual unitary) A matrix $U$ is dual unitary if $U$ and $U^R$ are unitary.
\end{definition}
\begin{definition}
(T-dual unitary) A matrix $U$ is called T-dual unitary if $U$ and $U^\Gamma$ are unitary.
\end{definition}
\begin{definition}
(2-unitary) A matrix $U$ is 2-unitary if it is dual unitary and T-dual unitary.
\end{definition}
Quantum circuit models constructed from dual unitaries have been widely studied recently as models of nonintegrable many-body quantum systems \cite{Bertini2019, claeys2020ergodic, Gopalakrishnan2019}, and the circuits constructed form 2-unitaries have been shown to possess extreme ergodic properties \cite{ASA_2021}. A generalization of 2-unitary matrices to multi-unitary matrices allows the construction of AME states with higher number of parties \cite{Goyeneche2015}.

The 2-unitary matrix corresponding to the state in Eq. (\ref{eq:ame_from_ols}) constructed from OLS gives a 2-unitary permutation defined as follows
\begin{align}
\begin{aligned}
P = \sum_{i,j} \ket{ij} \bra{K_{ij}L_{ij}}.
\end{aligned}
\end{align}
2-unitary permutations can be constructed from OLS in all local dimensions $d>2$, except $d=6$. It is also noted that if we multiply a 2-unitary permutation with a diagonal unitary matrix, it remains 2-unitary. In fact, permutations that are dual/T-dual unitary remain dual/T-dual unitary under the multiplication of all nonvanishing (unit) elements by phases--we refer to this as \textit{enphasing}. Thus, all 2-unitary permutations remain 2-unitary under enphasing.

\section{LU-equivalence of AME$(4,3)$ states}
This is the smallest case where four-party AME states exist. It has been shown that AME states of minimal support or rank 9 are all LU-equivalent \cite{Adam_SLOCC_2020}. It is now shown that 
\begin{thm}
There is only one LU-equivalent class of AME$(4,3)$ states or, equivalently, one LU-equivalent class of 2-unitary gates of size 9.
\end{thm}

\begin{proof}

A universal entangler on $\mathbb{C}^d \otimes \mathbb{C}^d$ entangles every product state, and it is known that they do not exist in $d=2$ and $d=3$ \cite{ChenDuan2007}.
 The fact that there are no two-qutrit universal entanglers implies that for any two-qutrit gate $U \in \mathbb{U}(9)$ there exists a product state in $\mathbb{C}^3 \otimes \mathbb{C}^3$ such that  
\beq
U (\ket{\alpha_1} \otimes \ket{\beta_1})=\ket{\alpha_2} \otimes \ket{\beta_2}.
\eeq 
Writing $\ket{\alpha_1} \otimes \ket{\beta_1}=(v_1 \otimes v_2)\ket{11}$ and $\ket{\alpha_2} \otimes \ket{\beta_2}=(u^\dagger_1 \otimes u^\dagger_2)\ket{11}$, where $u_i$ and $v_i$ are single-qutrit unitary gates, it is easy to see that 
\beq\label{eq:entangler_lu}
U_1=(u_1 \otimes u_2)U(v_1 \otimes v_2)
\eeq 
such that
\beq
U_1(\ket{1} \otimes \ket{1})=\ket{1} \otimes \ket{1}.
\eeq
Denoting non-zero entries of $U_1$ by $*$ , the matrix form of $U_1$ becomes
\beq
U_1=\left(\begin{array}{ccc|ccc|ccc}
1 & 0 & 0 & 0 & 0 & 0 & 0 & 0 & 0 \\
0 & * & * & * & * & * & * & * & * \\
0 & * & * & * & * & * & * & * & * \\
\hline
0 & * & * & * & * & * & * & * & * \\
0 & * & * & * & * & * & * & * & * \\
0 & * & * & * & * & * & * & * & * \\
\hline
0 & * & * & * & * & * & * & * & * \\
0 & * & * & * & * & * & * & * & * \\
0 & * & * & * & * & * & * & * & * \\
\end{array}\right).
\eeq

{ If $U$ is a 2-unitary operator, the LU tranformation in Eq.(\ref{eq:entangler_lu}) will lead to $U_1$ given in the matrix form in Eq.(\ref{eq:2_uni_gen_form}). The steps given below explains the reduction in the number of non-zero entries upon imposing dual unitary and T-dual unitary constraints.}
\begin{enumerate}
\item
Dual unitarity: If $U_1^R$ is unitary, then the nine $3 \times 3$ blocks in $U_1$ are orthonormal to each other. This implies that all nonzero entries in the first block containing $1$ vanish and in all the remaining blocks the element in the first row and the first column vanish. Therefore, the dual unitary constraint implies that $U_1$ is of the form 
\beq
U_1=\left(\begin{array}{ccc|ccc|ccc}
1 & 0 & 0 & 0 & 0 & 0 & 0 & 0 & 0 \\
0 & 0 & 0 & * & * & * & * & * & * \\
0 & 0 & 0 & * & * & * & * & * & * \\
\hline
0 & * & * & 0 & * & * & 0 & * & * \\
0 & * & * & * & * & * & * & * & * \\
0 & * & * & * & * & * & * & * & * \\
\hline
0 & * & * & 0 & * & * & 0 & * & * \\
0 & * & * & * & * & * & * & * & * \\
0 & * & * & * & * & * & * & * & * \\
\end{array}\right).
\label{eq:Udual_gen_form}
\eeq
This provides the most general form of the non-local part of two-qutrit dual-unitary gates.
\item
T-dual unitarity: If  $U_1^{\Gamma}$ is unitary, $U_1$ takes the form
\beq
U_1=\left(\begin{array}{ccc|ccc|ccc}
1 & 0 & 0 & 0 & 0 & 0 & 0 & 0 & 0 \\
0 & 0 & 0 & 0 & * & * & 0 & * & * \\
0 & 0 & 0 & 0 & * & * & 0 & * & * \\
\hline
0 & 0 & 0 & 0 & * & * & 0 & * & * \\
0 & * & * & * & * & * & * & * & * \\
0 & * & * & * & * & * & * & * & * \\
\hline
0 & 0 & 0 & 0 & * & * & 0 & * & * \\
0 & * & * & * & * & * & * & * & * \\
0 & * & * & * & * & * & * & * & * \\
\end{array}\right).
\label{eq:2_uni}
\eeq

Note that the four columns 2,3,4, and 7 have only four potentially non-vanishing elements. Thus these form a 4-dimensional orthonormal basis themselves and imply that the elements in the corresponding rows of columns 5,6,8, and 9 vanish. Therefore, Eq.~(\ref{eq:2_uni}) further simplifies to 
\beq
U_1=\left(\begin{array}{ccc|ccc|ccc}
1 & 0 & 0 & 0 & 0 & 0 & 0 & 0 & 0 \\
0 & 0 & 0 & 0 & * & * & 0 & * & * \\
0 & 0 & 0 & 0 & * & * & 0 & * & * \\
\hline
0 & 0 & 0 & 0 & * & * & 0 & * & * \\
0 & * & * & * & 0 & 0 & * & 0 & 0 \\
0 & * & * & * & 0 & 0 & * & 0 & 0 \\
\hline
0 & 0 & 0 & 0 & * & * & 0 & * & * \\
0 & * & * & * & 0 & 0 & * & 0 & 0 \\
0 & * & * & * & 0 & 0 & * & 0 & 0 \\
\end{array}\right).
\label{eq:2_uni_gen_form}
\eeq 
The LU transformation reduces the maximum number of non-zero entries of the 2-unitary from 81 to 33. 
\end{enumerate}

From $U_1$, we perform a series of local unitary transformations, involving only $2\times 2$ unitaries such that $  U_1 \stackrel{LU}{\sim} P_9 $, reducing the total number of non-zero entries to only 9.  Here $P_9$ is the 2-unitary permutation 
that takes $(11,12,13,21,22,23,31,32,33)$ to $(11,23,32,33,12,21,22,31,13)$  and results in the AME$(4,3)$ state $\ket{\Psi_{P_9}}$ in Eq. (\ref{eq:p9_state}). The details of the LU transformations can be found in the appendix \ref{app:lutransformations_theroem_1}.

We have shown that every 2-unitary matrix $U$ in $\mathbb{C}^3 \otimes \mathbb{C}^3$ can be expressed in the form 
$U=(u'_1\otimes u'_2)\, P_9\, (v'_1 \otimes v'_2)$. The four $3\times 3$ unitaries, $u'_{1,2}$ and $v'_{1,2}$ define a $33$-dimensional manifold of 2-unitaries.

Hence, it is proven that there is only one LU class of 2-unitary gates of size 9 or, equivalently, there is only one AME($4,3$) state up to multiplication by local unitaries.
\end{proof}
}{ 
\begin{corollary}
For any 2-unitary two-qutrit gate $U\in \mathbb{U}(9)$, there exist orthonormal product bases $\lbrace \ket{\alpha_{i} \beta_{j}} : i,j=1,2,3 \rbrace $ and $\lbrace \ket{\alpha'_{i} \beta'_{j}} :i,j=1,2,3 \rbrace $ in $\mathbb{C}^3 \otimes \mathbb{C}^3$ such that
\begin{align}\label{eq:product_basis}
\begin{aligned}
U \ket{\alpha_{i} \beta_{j}} = \ket{\alpha'_{i} \beta'_{j}},~~i,j=1,2,3.
\end{aligned}
\end{align}
\end{corollary}
Let the action of the 2-unitary permutation $P_9$ on the compuational basis be given by $P_9 \ket{ij} = \ket{k_{ij}l_{ij}}$, where $i,j,k_{i,j} ,l_{i,j} \in \lbrace 1,2,3 \rbrace$. The set of product states $\lbrace\ket{k_{i,j} l_{i,j}} : i,j=1,2,3\rbrace$ also form a product basis in $\mathbb{C}^3 \otimes \mathbb{C}^3$. Any 2-unitary two-qutrit gate can be expressed as $U=(u^\dagger_1\otimes u^\dagger_2)\, P_9\, (v^\dagger_1 \otimes v^\dagger_2)$, where $u_{1,2}$ and $v_{1,2}$ are $3\times 3$ unitaries. This allows to define the product states $\ket{\alpha_i \beta_j} =( v_1 \otimes v_2) \ket{ij}$ and $\ket{\alpha'_i \beta'_i} =( u^\dagger_1 \otimes u^\dagger_2) \ket{k_{ij}l_{ij}}$ which satisfies the condition in Eq.~(\ref{eq:product_basis}). If these local unitary operators $u_{1,2}$ and $v_{1,2}$ are used in the LU transformation in Eq.~(\ref{eq:entangler_lu}), the reduction to $P_9$ occurs in one step.
}
\subsubsection{Illustration of the proof of Theorem 1}
\label{sec:illustration}
We illustrate the proof of Theorem 1 for generic two-qutrit ($d^2=9$) 2-unitary matrices. Such generic 2-unitary operators have all $d^4=3^4$ elements nonzero and can be obtained from the algorithms presented in Ref.~\cite{SAA2020}. The absolute values of the entries of one such realization of a generic 2-unitary of size $9$, denoted simply as $U$ below is shown in Fig.~(\ref{fig:illustration}.a). We use an algorithm similar to the one described in Ref.~\cite{Shrigyan_2022} that, given an bipartite unitary $U$, finds a pair of maximally entangled states $\ket{\Phi_1}$ and  $\ket{\Phi_2}$ such that $U \ket{\Phi_1}=\ket{\Phi_2}$. We use the modified algorithm to find a product state that remains a product state under the action of $U$.  The existence of such a product state is guaranteed due to the non-existence of a two-qutrit universal entangler gate \cite{ChenDuan2007}. 

\begin{figure*}
\begin{tabular}{ccccc}
\includegraphics[scale=0.23]{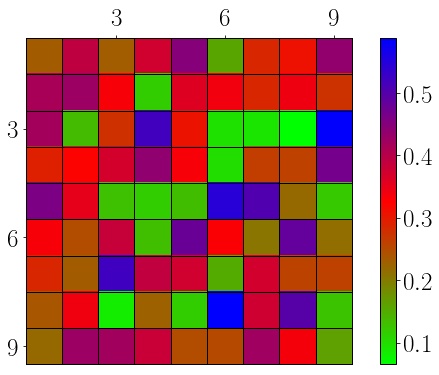}&
\includegraphics[scale=0.23]{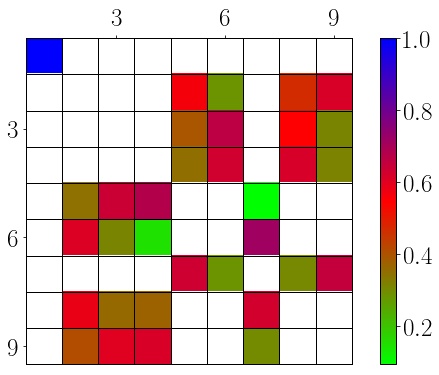}&
\includegraphics[scale=0.23]{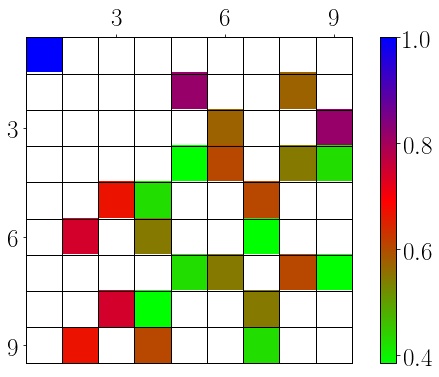}&
\includegraphics[scale=0.23]{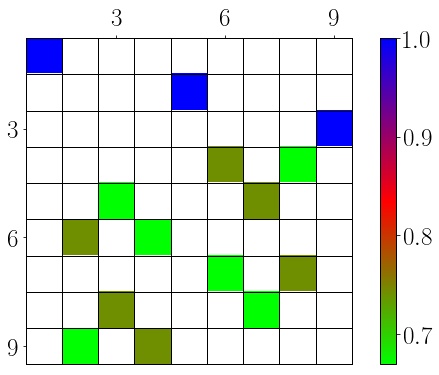}&
\includegraphics[scale=0.23]{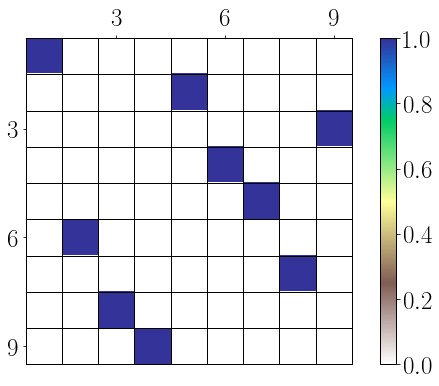}\\
(a)&(b)&(c)&(d)&(e)
\end{tabular}
\caption{(a) The absolute values of the entries of a generic 2-unitary of size $9$ obtained from the $M_{TR}$ algorithm presented in Ref.~\cite{SAA2020} are shown. All $d^4=81$  entries are nonzero and the rows and columns are entangled. The numbers $1$ to $9$ on the left indicate the row number and similarly, the numbers on the top indicate the column number (b) The absolute values of entries of the LU-equivalent 2-unitary matrix to $U$, $U_1=(u_1^{\dagger} \otimes v_1^{\dagger}) U (u_2 \otimes v_2)$. The only nonzero entry in the first block is equal to 1 and the number of nonzero entries in the 2-unitary matrix is $81-48=33$. (c,d) The absolute values of entries of the 2-unitary matrix are shown after performing local transformation defined in Eq.~(\ref{eq:step_2}) and Eq.~(\ref{eq:step_3}). (e) The local transformation defined in Eq.~(\ref{eq:step_4}) results in 2-unitary enphased permutation in which there is only one nonzero entry (of modulus $1$) in any row and column.}
\label{fig:illustration}
\end{figure*}

Using the product state obtained from the algorithm for the 2-unitary $U$, we illustrate the proof of Theorem 1 step-by-step for $U$ showing that it is LU-equivalent to the $P_9$ (2-unitary permutation) as follows:

{\blu Let $\ket{\Psi_X}$ and $\ket{\Psi_Y}$ be the product states found using the modified algorithm such that $U \ket{\Psi_X} = \ket{\Psi_Y} $. Define matrices $X$ and $Y$ with elements $X_{ij} = \braket{ij|\Psi_X}$ and $Y_{ij} = \braket{ij|\Psi_Y}$, respectively, where $i,j=1,2,...,d$. Let the singular value decompositions of the matrices $X$ and $Y$ are given by
\begin{align}
\begin{aligned}
X=a_1 D_1 b_1,~~~Y = a_2 D_2 b_2,
\end{aligned}
\end{align}
where $a_1,a_2,b_1,b_2$ are $d \times d$ unitary matrices and $D_1,D_2$ are diagonal matrices with only one nonzero entry that is equal to $1$. Rewriting the above equations in the vector form gives $\ket{\Psi_X} = (a_1 \otimes b_1^T) \ket{11}$ and $\ket{\Psi_Y} = (a_2 \otimes b_2^T) \ket{11}$. The unitary matrices obtained from these relations can be used to implement the LU transformation given in Eq.(\ref{eq:entangler_lu}) by setting $u_1 = a_1, u_2 = b_1^T, v_1 = a^\dagger_2$ and $v_2 = b_2^*$. 
The transformation results in a 2-unitary matrix whose nonzero entries are shown in Fig.~(\ref{fig:illustration} b). } Note that the number of nonzero entries is reduced from $81$ to $33$. The non-existence of a two-qutrit universal entangler together with 2-unitarity conditions imply that any 2-unitary of size $9$ is of the form shown in Fig.~(\ref{fig:illustration} b). The matrices $U_1^R$ and $U_1^{\Gamma}$ also have exactly the same structure for the non-zero entries.

We apply a sequence of local unitary transformations defined in Eqs.~(\ref{eq:step_2})--(\ref{eq:step_4}) to further simplify the 2-unitary matrix. The result of these local transformations is shown in Figs.~(\ref{fig:illustration} c)--~(\ref{fig:illustration} e). 
{ 
\section{Computable and complete set of LU-invariants}
In this section, we provide a complete set of LU-invariants that, in principle, allow us to determine if two given operators are LU-equivalent.  A family of LU-invariants for a general matrix $A \in \mathbb{C}^d \otimes \mathbb{C}^d$ may be constructed as follows. Take any natural number $n$ and permutations $\sigma,\tau,\rho,\lambda \in S_n$ - the symmetric group on $\{1,\cdots,n\}$. From this data compute the following number: $A( \sigma,\tau,\rho,\lambda) =$
\beq
\label{eq:LU_inv}
 A^{i_1j_1}_{k_1l_1}  \cdots A^{i_nj_n}_{k_nl_n} \;
(A^{\dagger})^{k_{\rho(1)}l_{\lambda(1)}}_{i_{\sigma(1)}j_{\tau(1)}} \cdots (A^{\dagger})^{k_{\rho(n)}l_{\lambda(n)}}_{i_{\sigma(n)}j_{\tau(n)}},
\eeq
where the sum over repeated indices is assumed.
It is straightforward to check that this is an invariant of LU-equivalence. The content of Propositions 8 and 20 of Ref.~\cite{VijayK} is that the collection of all these numbers is a complete LU-invariant. In other words, $A, B $, not necessarily unitary, are LU-equivalent, if and only if $A(\sigma,\tau,\rho,\lambda) = B(\sigma,\tau,\rho,\lambda)$ for every $n$ and choice of $\sigma,\tau,\rho,\lambda \in S_n$. In the present work, this is used to show that matrices $A$ and $B$ are not LU-equivalent by displaying $\sigma,\tau,\rho,\lambda \in S_n$ for some $n$ such that $A(\sigma,\tau,\rho,\lambda) \neq B(\sigma,\tau,\rho,\lambda)$. 

In general, for the LU-invariant $A(\sigma,\tau,\rho,\lambda)$ to be nontrivial and not obtainable from a smaller value of $n$, the four permutations of length $n$ must form a $4 \times n$ Latin rectangle {i.e., there are $n$ different symbols in all four rows and $4$ different symbols in all $n$ columns. For $n<4$, there are repetitions of symbols and the resulting invariant reduces to a function of either trivial invariant $\Tr A A^{\dagger}$ or known invariants based on matrix rearrangements such as realignment and partial transpose.}

{Our main interest is to distinguish 2-unitary operators that are not LU equivalent. The problem of LU-equivalence for 2-unitaries is specially hard because all known LU invariants based on matrix rearrangements such as realignment and partial transpose are constants \cite{Ma2007,Bhargavi2017}. For distinguishing 2-unitaries that are not LU equivalent, we need to choose the four permutations in Eq.~\ref{eq:LU_inv} in such a way that the resulting invariant does not reduce to a function of known invariants involving realignment and partial transpose matrix rearrangements. A possible choice of such permutations for $n=4$ is 
\beq
\label{eq:perm_ind}
\begin{split}
&\sigma=(1\,2\,3\,4), \, \tau=(2\,1\,4\,3), \\ & \rho=(3\,4\,1\,2), \, \lambda=(4\,3\,2\,1).
\end{split}
\eeq
Note that  the above four permutations arranged in a $4\times4$ arrangement form a Latin square of size 4 and one of them can be chosen to be identity. The resulting LU-invariant is equal to
\beq
\label{eq:Invar_n_4}
A^{i_1j_1}_{k_1l_1}  A^{i_2 j_2}_{k_2 l_2}  A^{i_3 j_3}_{k_3l_3}  A^{i_4 j_4}_{k_4l_4} 
 (A^{\dagger})^{k_3 l_4}_{i_1  j_2}   (A^{\dagger})^{k_4 l_3}_{i_2  j_1}  (A^{\dagger})^{k_1 l_2}_{i_3  j_4}  (A^{\dagger})^{k_2 l_1}_{i_4  j_3} ,
\eeq
and is useful in distingushing  2-unitaries that are not LU equivalent as discussed below.
 }
}
\section{ LU-equivalence classes of AME$(4,d)$ states in $d\geq 4$}
{ In this section, using the concept of LU-invariants, we study the LU-equivalence classes of AME states in $d\geq 4$.}

\begin{thm}
The number of LU-equivalence classes of 2-unitary  gates of size $d^2$ (equivalently, the number of LU-equivalence classes of AME states of four qudits) for $d\geq4$ is infinite.
\end{thm}

\begin{proof} 
{
Consider first the case of $d \geq 4$ and $d \neq 6$.  
We observed that 2-unitary permutations remain 2-unitary under enphasing-- multiplication of all nonvanishing (unit) elements by phases. However, such 2-unitaries are not necessarily LU equivalent.  
}

What we actually see is that given one permutation gate, there are infinitely many LU-equivalence classes of enphased permutation gates of the same size, the method of proof necessitating the restriction $d \neq 6$ - since permutation 2-unitary gates (which are in bijection with orthogonal Latin squares) are known to exist for all $d \geq 4$ except for $d=6$.

Fix a permutation $P$ of size $d^2$. For every $i,j \in [d] = \{1,2,\cdots,d\}$, there exist unique $k,l \in [d]$ such that $P^{ij}_{kl} = 1$. 
Let $S \subseteq [d]^{\times 4}$ be the set of all $(i,j,k,l)$ as $i,j$ vary over $[d]$ and $k,l$ are the unique elements with 
$P^{ij}_{kl} = 1$.
For $t=1,2,3,4$, let $\pi_t:S \rightarrow [d]$ be the $t^{\text{th}}$ component function, so that, for instance, $\pi_3((i,j,k,l)) = k$.
Note that $|S| = d^2$ and that for each $k \in [d]$, there are exactly $d$ elements  $s \in S$ with  $\pi_3(s) =k$, and similarly for 
each $l \in [d]$, there are exactly $d$ elements $s \in S$ with  $\pi_4(s) = l$.

We will be interested in multi-subsets of $S$, consisting of sets that have elements of $S$ that could be repeated. Such multi-subsets construct the invariant in Eq.~(\ref{eq:LU_inv}) and define functions $X(s)$ from $S$ to $\{0,1,2,3,\cdots\}$  counting the number of times $s$ occurs in $X$. Any such multi-subset $X$ also determines four functions from $[d]$ to $\{0,1,2,3,\cdots\}$. These are $I_X,J_X,K_X,L_X$ where 
$$I_X(p) = \sum_{s \in S, \pi_1(s) = p} X(s),$$ with analogous definitions for $J_X,K_X$ and $L_X$. These count how many times $p$ occurs as a first, second, third, or fourth component of elements of X.

We claim that there are two distinct multi-subsets $X,Y$ of $S$ for which all these functions are identical. To see this, note that a multi-subset $X$ of $S$ corresponds naturally to a function $F: [d] \times [d] \rightarrow \{0,1,2,3,\cdots\}$. For given such a function, we could define $X: S \rightarrow  \{0,1,2,3,\cdots\}$ by $X((i,j,k,l)) = F(i,j)$ and conversely, given $X$, we may define $F(i,j) = X((i,j,k,l))$ where  $k,l$ are the unique elements with $P^{ij}_{kl} = 1$. 

Say $X$ corresponds to $F$ and $Y$ to $G$.
The condition that $I_X = I_Y$ is given by $\sum_j F(i,j) = \sum_j G(i,j)$, for each $i \in [d]$. Similarly the condition that $J_X = J_Y$ is given by $\sum_i F(i,j) = \sum_i G(i,j)$, for each $j \in [d]$. The condition that $K_X = K_Y$ is not as easily expressed since it depends on the permutation $P$, but it is clear that is given by a sum of $d$ $F(i,j)$s
in the LHS and the corresponding $G(i,j)$s
in the RHS where the $(i,j)$ vary over those for which the corresponding $k$'s are equal, for each $k \in [d]$. A similar statement holds for  when
$L_X = L_Y$. 

To summarise, two multi-subsets $X,Y$ corresponding to functions $F,G$ have the same $I,J,K,L$ functions exactly (the $K$ and $L$ functions should not be confused with the Latin square symbols) when $4d$ homogeneous linear equations in $F(i,j) - G(i,j)$ are satisfied. However these equations are not independent because the sum of all the $F(i,j)$ coincides with the sum of all $G(i,j)$ once $I_X = I_Y$ and so we need to consider only $d-1$ equations for each of the $J,K$, and $L$. The actual number of equations is thus at most $4d-3$. The number of variables is $d^2$, namely the
$F(i,j) - G(i,j)$. 
For $d \geq 4$, $d^2 > 4d-3$, that is, the number of variables is greater than the number of equations. Since it is a system of homogeneous equations with more variables than equations, these equations have a non-trivial solution. The coefficients in the system of equations are rational, and hence a rational solution exists. Further, by clearing all the denominators by
multiplying by some number, we may assume that this solution is integral and then choose each $F(i,j)$ and $G(i,j)$ to be non-negative. The homogeneity of the equations ensures that this remains a solution.

Finally, we have two distinct multi-subsets $X,Y$ of $S$ such that for any $p$, the number of elements of $X$ with first component $p$ equals the number of elements of $Y$ with first component $p$, and similarly for the other three components too. Say these multi-subsets have $N$ elements each. This condition implies that there exist permutations
$\sigma, \tau,\rho,\lambda \in S_N$ such that if $X$ is enumerated (arbitrarily) as $\{(i_1,j_1,k_1,l_1), \cdots, (i_N,j_N,k_N,l_N)\}$
then $Y = \{(i_{\sigma(1)},j_{\tau(1)},k_{\rho(1)},l_{\lambda(1)}), \cdots, (i_{\sigma(N)},j_{\tau(N)},k_{\rho(N)},l_{\lambda(N)})\}$. Note that $N=\sum_{i,j}F(i,j)=\sum_{i,j}G(i,j)$.

Since $X$ and $Y$ are distinct, there is an $s=(i,j,k,l) \in S$ for which $X(s) \neq Y(s)$. Let $Q$ be the 2-unitary obtained from a 2-unitary permutation $P$ by setting $Q^{ij}_{kl} = \alpha \in S^1$ with other entries untouched. The invariant $Q(\sigma,\tau,\rho,\lambda)$, by definition, is a sum of terms each of which is a monomial in $\alpha,\overline{\alpha}$. It suffices to see that one of these terms is a non-zero power of $\alpha$  for this polynomial to take infinitely many values as $\alpha$ ranges over $S^1$. But this is true because the term corresponding to $\{(i_1,j_1,k_1,l_1), \cdots, (i_N,j_N,k_N,l_N)\}$ is $\alpha^{X(s)}\overline{\alpha}^{Y(s)}$ which is a non-zero power of $\alpha$.
\end{proof}

To illustrate that enphasing of 2-unitary permutations leads to different LU-equivalence classes in { $d = 4$}, consider the permutation denoted $P_{16}$ and $(P_{16})^{ij}_{kl}$ elements are such that $(kl)$ are ordered as $(11,44,22,33,43,12,34,21,24,31,13,42,32,23,41,14)$ when $(ij)$ is in the lexicographic ordering $(11,12, \cdots, 44)$. Let $P_{16}(\theta)$ denote the 2-unitary obtained from $P_{16}$ by changing only $(P_{16})^{11}_{11}$ from $1$ to $e^{i \theta}$.

{The invariant given in Eq.~\ref{eq:Invar_n_4} for} $P_{16}(\theta)$ evaluates to the following simple continuous function of $\theta$:
\beq
\label{eq:P16_LUI}
P_{16}\left(\theta; \sigma,\tau,\rho,\lambda\right)=8(29+3\cos \theta).
\eeq 
 As $ \theta $ ranges in $(-\pi,\pi]$, the invariant takes infinitely many distinct values and so the corresponding 2-unitaries are not LU-equivalent.

While an interpretation of this invariant is not clear, it can be related to a moment of an operator on two copies, or four parties $A,B,C,$ and $D$ with $U$ acting on the pairs $(A,B)$ and $(C,D)$:
\beq
L[U]:=(S_{BD} \otimes \mathbb{I}_{AC}) (U^{\dagger} \otimes U^{\dagger})(S_{BD} \otimes \mathbb{I}_{AC}) (U \otimes U),
\label{eq:L_U_op}
\eeq
where $S_{BD}$ is the {\sc SWAP} gate between subsystems $B$ and $D$. { Representing bipartite operator with a tensor having two incoming and two outgoing indices, diagrammatic representation of $L[U]$ in terms of bipartite unitary operators $U$ and $SUS$ (where $S$ is the SWAP gate) is given by
 \beq
\includegraphics[scale=0.3]{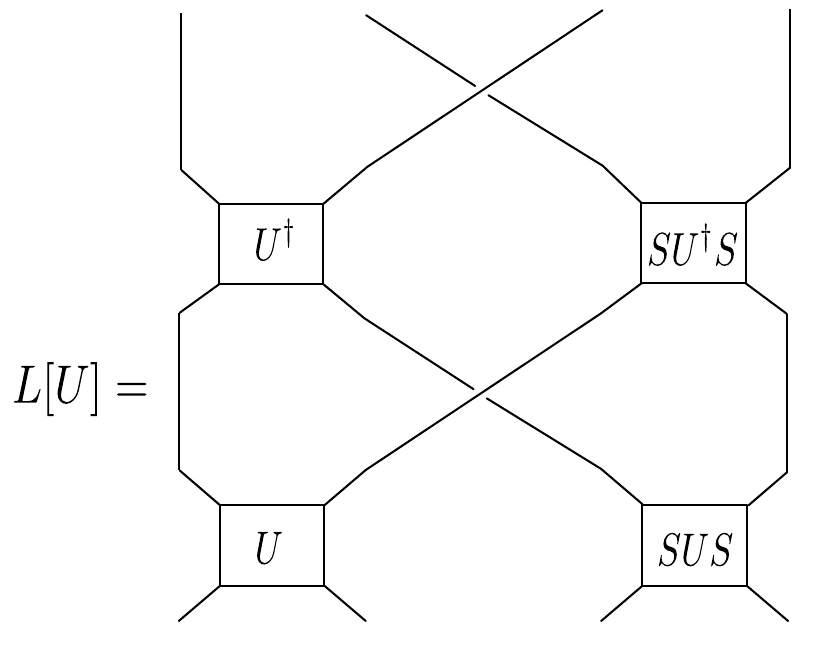}
\eeq

It can be easily checked that $L[U]=L[SUS]$ and all the moments, $\Tr L^k[U];k=1,2,3,\cdots$, are local unitary invariants. For example,} $\Tr  L[U]= \Tr (U^R U^{R\dagger})^2$ is related to the operator entanglement of $U$ \cite{Zanardi2001,Zyczkowski2004,SAA2020} and the invariant in Eq.~(\ref{eq:P16_LUI}) is equal to the second moment $\Tr  L^2[P_{16}(\theta)]$. Note that $\Tr L[U]$ is equal to $d^2$ for dual unitaries and thus cannot distinguish dual unitaries in different LU-equivalence classes.

\subsection{Proof based on Orthogonal Diagonal Latin Squares}\label{ODLS}

In this section, we give a constructive proof for Theorem 2 for the case $ d\geq 4 $ and $ d \neq 6 $ using special orthogonal Latin squares. { \blu We find explicit examples of multi-sets with desired properties discussed above, and construct the four permutations $\sigma,\tau,\rho,\lambda$}.

Consider a Latin square of order $ d $ with elements from the set $[d]= \lbrace1,2,..,d \rbrace $. {\blu A transversal of a Latin square is a set of $d$ distinct entries such that no two entries share the same row or column.} A diagonal Latin square is one in which both the main diagonal and the main back (or ``anti-") diagonal are transversals.

Two Latin squares $ \lK $ and $ \lL $ form a pair of orthogonal diagonal Latin squares (ODLS) if both are diagonal Latin squares and orthogonal.  It is known that ODLS's exist for every order $ d $ except 2,3 and 6 \cite{brown2017completion}. 
An example of a pair of orthogonal diagonal Latin squares in $ d=4 $ is given below
\begin{align}
\label{eq:ODLS1}
\begin{aligned}
ODLS(4) = \begin{array}{|c|c|c|c|}
\hline 
1,2&3,3&4,1&2,4\\
\hline
4,4&2,1&1,3&3,2\\
\hline
2,3&4,2&3,4&1,1\\
\hline
3,1&1,4&2,2&4,3\\
\hline 
\end{array}
\end{aligned}.
\end{align}
Given a pair of orthogonal Latin squares $ \lK $ and $ \lL $, we can construct a 2-unitary operator $ P $ as follows
\begin{align}
\begin{aligned}
P = \sum_{i,j=1}^d \alpha_{i,j} \ket{i,j} \bra{\lK_{i,j}, \lL_{i,j}},
\end{aligned}
\end{align} 
where $ \alpha_{i,j} \in S^{1} $. { \blu For our purpose, we set $  \alpha_{i,j} =1,$  for all $i,j$ except $ \alpha_{1,1}=\alpha \neq 1$.}

We choose both $ \lK $ and $ \lL $ to be diagonal Latin squares such that they form a pair of orthogonal diagonal Latin squares. Let $ S \subseteq [d]^{\times 4}$ be the set of all $ (i,j,\lK_{i,j},\lL_{i,j}) , ~i,j \in [d]$ such that $ P^{i,j}_{\lK_{i,j},\lL_{i,j}}$ is non-zero. It is noted that the multi-subsets of $ S $ construct the invariant in Eq.~(\ref{eq:LU_inv}).
{ \blu Consider the following subsets of $ S $ constructed using the addresses and elements of main and back diagonals of $\lK$ and $\lL$}:
\begin{align}\label{eq:multisets}
\begin{aligned}
X &= \lbrace (i,i,\lK_{i,i},\lL_{i,i})  , ~i\in [d]\rbrace,\\
Y &= \lbrace (i,d+1-i,\lK_{i,d+1-i},\lL_{i,d+1-i}) ,~ i\in [d] \rbrace.\\
\end{aligned}
\end{align}
An element $ (i,i,\lK_{i,i},\lL_{i,i})  \in X $ is different from any element  $ (j,d+1-j,\lK_{j,d+1-j},\lL_{j,d+1-j})  \in Y $ except the case when $ d $ is odd and $ i=j=(d+1)/2 $. Therefore, $ X $ and $ Y $ are distinct. However, for the multi-subset $X$ the functions $I_X,J_X,K_X,L_X$ are identical with the corresponding functions for the multi-subset $Y$, both being constant functions (equal to 1).  For the example in Eq.~(\ref{eq:ODLS1}), $X=\{(1,1,1,2),(2,2,2,1),(3,3,3,4),(4,4,4,3)\}$ and $Y=\{(1,4,2,4),(2,3,1,3),(3,2,4,2),(4,1,3,1)\}.$ 

The four permutations $\sigma,\tau, \rho,\lambda$ can be found by inspection since all the functions $I,J,K,L$ evaluates to 1 for the subsets $X$ and $Y$. Note that the sets $ \lbrace \lK_{i,i},~i \in[d]\rbrace$  and $\lbrace \lK_{i,d+1-i},~i \in[d]\rbrace $ contains elements in the main diagonal and the back diagonals of $ \lK $, respectively. Since $ \lK $ is assumed to be a diagonal Latin square, these two sets are related by a permutation. This gives the permutation $ \rho \in S_d $. A similar argument can be given in the case of sets  $ \lbrace \lL_{i,i},~i \in[d]\rbrace$  and $\lbrace \lL_{i,d+1-i},~i \in[d]\rbrace $, and the corresponding permutation is $ \lambda \in S_d$. It is also evident that the set $ [d] $ and $ \lbrace d+1-i  , ~i\in[d]\rbrace$ are related by a permutation
\begin{align}
\begin{aligned}
\tau = \begin{pmatrix}
1&2&...&d\\
d&d-1&...&1
\end{pmatrix}.
\end{aligned}
\end{align}

Therefore, if we enumerate elements in $ X $ as $ \lbrace (i_1,j_1,k_1,l_1),\cdots, (i_d,j_d,k_d,l_d)\rbrace $, then $ Y = \lbrace (i_1,j_{\tau(1)},k_{\rho(1)},l_{\lambda(1)}),\cdots, (i_d,j_{\tau(d)},k_{\rho(d)},l_{\lambda(d)}) \rbrace $. The permutation $ \rho $ is identity in this case.
Note that the element $ (1,1,\lK_{1,1},\lL_{1,1})  \in X$ does not belong to $ Y $. Then, the term corresponding to $ \lbrace (i_1,j_1,k_1,l_1),\cdots, (i_d,j_d,k_d,l_d)\rbrace $ in the LU invariant $ P( 1, \tau,\rho,\lambda) $ evaluates to $ \alpha $. Therefore $ P( 1, \tau,\rho,\lambda) $ can have infinitely many values as $ \alpha $ is a continuous parameter. Hence it shows that there exists an infinite number of LU-equivalence classes of 2-unitary gates for $ d \geq4 $ except $ d=6 $. {\blu In $d=3$, as proven earlier, there is only one LU-equivalence class of AME(4,3) states. This is consistent with the fact that there are no ODLS in $d=3$ \cite{brown2017completion}.}

\subsection{ Special case of $d=6$:}
Due to the non-existence of orthogonal Latin squares of size 6 \cite{bose1960further}, 2-unitary permutations of size 36 do not exist \cite{Clarisse2005} and we need to treat this case separately. However, it was shown recently in Ref.~\cite{SRatherAME46} that a 2-unitary matrix of size 36 denoted as $\mathcal{U}_{36}$, or, equivalently, AME state of four six-level systems, AME$(4,6)$, $\ket{\Psi_{\mathcal{U}_{36}}}$ exists. This settled positively a long-standing open problem in quantum information theory \cite{Open2020}. 

The number of non-zero elements of $\mathcal{U}_{36}$, equivalently, coefficients of $\ket{\Psi_{\mathcal{U}_{36}}}$ is $112$, and involve the $20^{\text{th}}$ root of unity $ \omega=\exp(2 \pi i /20)$ and the real numbers \cite{SRatherAME46}:
\beq
\label{abc}
a=\left(5+\sqrt{5}\right)^{-1/2},\, b/a=\varphi, \, c=1/\sqrt{2},
\eeq
where $\varphi=(1+\sqrt{5})/2$ is the golden ratio. For the sake of completeness, we show the nonzeros matrix elements of the 2-unitary $\mathcal{U}_{36}$ corresponding to the {\em golden} AME($4,6$) state \cite{SRatherAME46} in Fig.~(\ref{fig:Table_AME_46}). The pair of indices $(k,l)$ shown in rows label the rows and the $(i,j)$  shown in columns label the columns in $\mathcal{U}_{36}$. 

\begin{figure*}
\begin{tabular}{c}
\includegraphics[scale=0.23]{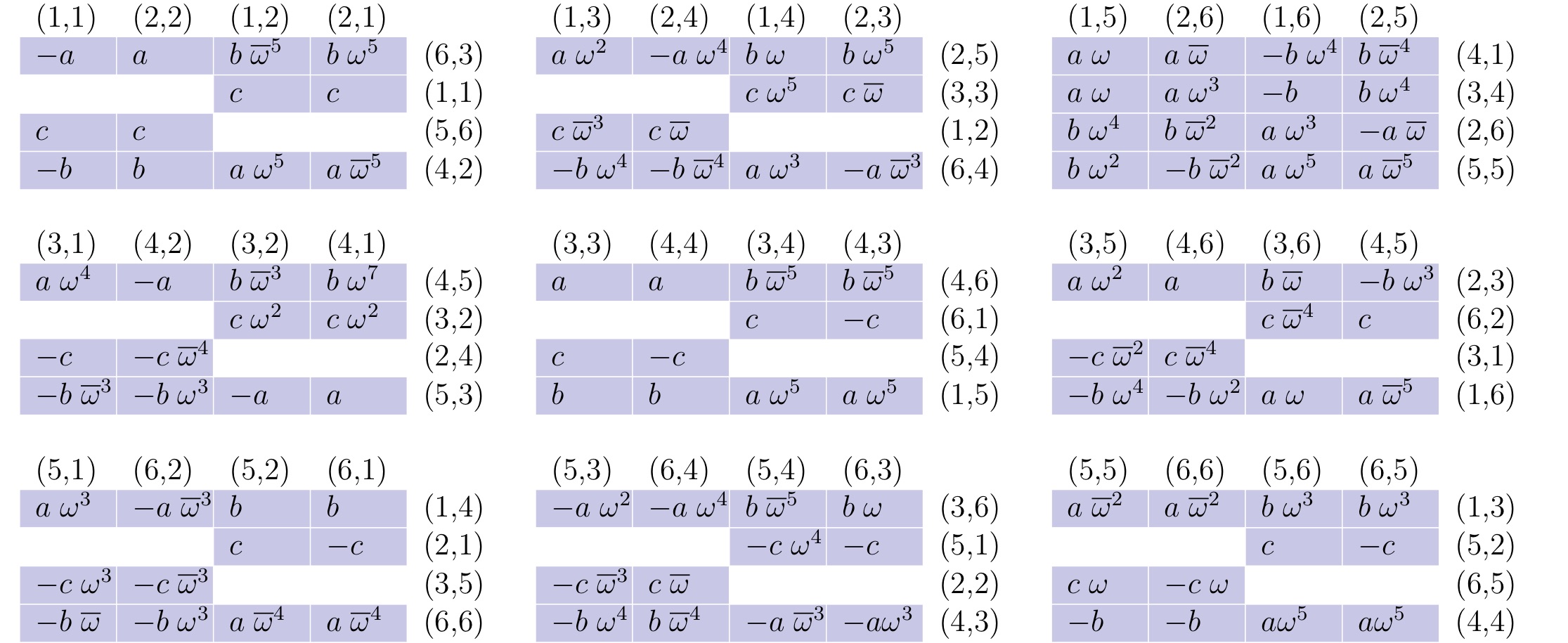}
\end{tabular}
\caption{Non-vanishing matrix elements of the 2-unitary $\mathcal{U}_{36}$. The pair of indices $(i,j)$ indicated in rows label the relevant row, $j+6(i-1)$, and the pair of indices $(k,l)$ indicated in columns determine the relevant column, $l+6(k-1)$. $\overline{\omega}$ denotes the complex conjugate of $\omega$.}
\label{fig:Table_AME_46}
\end{figure*}

Here, we show that one can obtain an infinite number of 2-unitaries from $\mathcal{U}_{36}$ by multiplying it with appropriate diagonal unitaries. Unlike 2-unitary permutations,  $ \mathcal{U}_{36}$ does not remain 2-unitary under multiplication by diagonal unitaries with arbitrary phases. In order to preserve 2-unitarity, one needs to multiply particular rows or columns of the given 2-unitary with specially designed phases depending on its structure. The simplest example in the case of $\mathcal{U}_{36}$ is the one-parameter family of 2-unitaries
\beq
\label{eq:one_param_U_36}
\mathcal{U}_{36} (\theta)=\mathcal{D}(\theta)\, \mathcal{U}_{36},
\eeq
where $
\mathcal{D}(\theta)=\text{Diag}[e^{i \theta},1,1,e^{i \theta},1,1,e^{i \theta},1,1,e^{i \theta},1_{26}]$
and $\theta \in (-\pi, \pi]$. The notation $1_{K}$ is to indicate a length- $K$ string of $1$s.
This yields an infinity of AME$(4,6)$ states parameterized by $\theta$ under the correspondence in Eq.~(\ref{eq:AME_U}). That this enphasing retains the 2-unitary property is not evident, but follows from the observation that 
$[\mathcal{U}_{36}(\theta)]^R= \mathcal{U}_{36}^R \, \mathcal{D}'(\theta)$ and $[\mathcal{U}_{36}(\theta)]^{\Gamma}= \mathcal{D}''(\theta) \mathcal{U}_{36}^{\Gamma}$, where $\mathcal{D}'(\theta)=\text{Diag}[e^{i \theta}, e^{i \theta}, 1_{16},e^{i \theta},e^{i \theta},1_{16}]$, and $\mathcal{D}''(\theta)=\text{Diag}[e^{i \theta}, e^{i \theta}, 1,1,1,1,e^{i \theta},e^{i \theta},1_{28}] $.

We show that these 2-unitaries are not LU-equivalent by evaluating the invariant $\mathcal{U}_{36} (\theta; \sigma,\tau,\rho,\lambda)$ for permutation of indices given by Eq.~(\ref{eq:perm_ind}). The invariant evaluates to the following function of $\theta$:
\beq
\label{eq:U36_LUI}
\mathcal{U}_{36}(\theta; \sigma,\tau,\rho,\lambda)=C_0 +6 \cos \theta,
\eeq
where $C_0=3\left(202 + \sqrt{5} + 2\sqrt{5 - 2 \sqrt{5}}\right)/4 \approx 154.267$ is a constant. 
 The $\theta$ dependence proves that the invariant can take infinitely many distinct values and the corresponding $\mathcal{U}_{36}(\theta)$ are not LU-equivalent.The invariant in Eq.~(\ref{eq:U36_LUI}) is equal to $\tr L^2[\mathcal{U}_{36}(\theta)]=C_0 +6 \cos \theta,$.
The realignment and partial transpose of 2-unitaries provide other 2-unitaries and for $\mathcal{U}_{36}(\theta)$, one needs to evaluate the third moment $\tr  L^3[U]$ to show that these are not LU-equivalent. 
\subsection{25-parameter family of AME$(4,6)$ states}
{  Apart from the one-parameter family of enphasing discussed above, we give a more general construction consisting of $25$ real parameters.} Let ${\mathcal U} = {\mathcal U}_{36}$. This has 112 non-zero entries. 
Consider a matrix, say ${\mathcal V}$, obtained from ${\mathcal U}$ by 
multiplying each of these non-zero entries by a phase factor. For 
definiteness, suppose that
each non-zero ${\mathcal U}_{ij}$ is multiplied by $e^{i\theta_{ij}}$. 
We now try to understand under what conditions on the $\theta_{ij}$ is 
the new matrix ${\mathcal V}$ also 2-unitary.

First ${\mathcal V}$ must be unitary. Since its rows are still of norm 
1, only the orthogonality of the rows needs to be ensured. Take any two 
rows say $i_1,i_2$ of ${\mathcal V}$. Suppose that the columns where 
both these rows have non-zero entries are $j_1,j_2,...$ (at most 4, 
from the structure of ${\mathcal U}$). For the inner-product of these 
rows to vanish it suffices that $\theta_{i_1j_1} - \theta_{i_2j_1} = 
\theta_{i_1j_2} - \theta_{i_2j_2} = \theta_{i_1j_3} - \theta_{i_2j_3} 
= ... $. This is a set of homogeneous linear equations in the 
$\theta_{ij}$. Similarly for ${\mathcal V}^R$ and ${\mathcal 
V}^\Gamma$ to be unitary we get other homogeneous linear equations in 
the $\theta_{ij}$.

Writing out all these homogeneous linear equations for the 
$\theta_{ij}$, we get a system of 246 equations - 75 for ${\mathcal 
V}$, 87 for ${\mathcal V}^\Gamma$ and 84 for ${\mathcal V}^R$ - in 112 
variables. The rank of the coefficient matrix can be computed to be 87 
using, say Mathematica, thereby yielding a 25-dimensional solution 
space. This is the required 25-dimensional family of 2-unitary 
enphasings of  ${\mathcal U}$.

Apart from solving the difficult problem of establishing LU-equivalence classes for AME states of four parties or 2-unitary operators in any local dimension, the methods developed herein can be extended both to unequal local dimensions and to more parties. This requires as many permutations as the number of parties to construct the invariants. That the case of qutrits are special and have only one class needs further elucidation in terms of the geometry of the set of 2-unitaries in this case. We hope that these results pave for a deeper understanding of multipartite states and new entanglement measures.

\begin{acknowledgments}
NR acknowledges funding from the Center for Quantum Information Theory in Matter and Spacetime, IIT Madras. Funding support from the Department of Science and Technology, Govt. of India, under Grant No. DST/ICPS/QuST/Theme-3/2019/Q69 and MPhasis for supporting CQuiCC are gratefully acknowledged. We thank S. Aravinda for initiating discussions on the topic.
\end{acknowledgments}

\bibliographystyle{unsrt}
\bibliography{ent_local}

\onecolumngrid

\appendix

\section{LU transformations involved in proof of the theorem 1}\label{app:lutransformations_theroem_1}
{
There exist no universal entanglers in $d=3$. This result allows us to find a unitary operator $U_1$ that is LU-equivalent to a given two-qutrit operator $U $ such that the entry in the first column and first row of $U_1$ is equal to 1. The corresponding LU transformation, given in Eq. (\ref{eq:entangler_lu}), is restated here for completeness:
\beq
U_1=(v_1^{\dagger} \otimes v_2^{\dagger})U(u_1 \otimes u_2),
\eeq 
where $u_{1,2}$ and $v_{1,2}$ are single-qutrit unitary gates.

Requiring $U$ to be 2-unitary and imposing 2-unitarity constraints will lead to the following matrix form of $U_1$:
\beq
U_1=\left(\begin{array}{ccc|ccc|ccc}
1 & 0 & 0 & 0 & 0 & 0 & 0 & 0 & 0 \\
0 & 0 & 0 & 0 & * & * & 0 & * & * \\
0 & 0 & 0 & 0 & * & * & 0 & * & * \\
\hline
0 & 0 & 0 & 0 & * & * & 0 & * & * \\
0 & * & * & * & 0 & 0 & * & 0 & 0 \\
0 & * & * & * & 0 & 0 & * & 0 & 0 \\
\hline
0 & 0 & 0 & 0 & * & * & 0 & * & * \\
0 & * & * & * & 0 & 0 & * & 0 & 0 \\
0 & * & * & * & 0 & 0 & * & 0 & 0 \\
\end{array}\right).
\eeq

From here, we apply appropriate local unitary transformations to simplify further.}

In the following step, we label relevant non-zero entries and represent $ U_1 $ as
\begin{align}
\begin{aligned}
U_1=\left(\begin{array}{ccc|ccc|ccc}
1 & 0 & 0 & 0 & 0 & 0 & 0 & 0 & 0 \\
0 & 0 & 0 & 0 & p_{11} & p_{12} & 0 & q_{11} & q_{12} \\
0 & 0 & 0 & 0 & p_{21} & p_{22} & 0 & q_{21} & q_{22} \\
\hline
0 & 0 & 0 & 0 & * & * & 0 & * & * \\
0 & * & * & * & 0 & 0 & * & 0 & 0 \\
0 & * & * & * & 0 & 0 & * & 0 & 0 \\
\hline
0 & 0 & 0 & 0 & * & * & 0 & * & * \\
0 & * & * & * & 0 & 0 & * & 0 & 0 \\
0 & * & * & * & 0 & 0 & * & 0 & 0 \\
\end{array}\right).
\end{aligned}
\end{align}
Consider the matrices
\begin{align}
\begin{aligned}
P = \left(\begin{array}{cc} p_{11} & p_{12}\\ p_{21} & p_{22} 
\end{array} \right), ~~ Q = \left(\begin{array}{cc} q_{11} & q_{12}\\ q_{21} & q_{22} 
\end{array} \right).
\end{aligned}
\end{align} 
The constraint that $ U_1 $ be 2-unitary gives the following conditions:
\begin{align}
\begin{aligned}
&P P^\dagger + Q Q^\dagger = \mathbb{I}_2 \; \text{(Unitarity)}\\& P^\dagger P+  Q^\dagger Q= \mathbb{I}_2\; \text{(T-dual)},\\
&\left.
\begin{aligned}
&\tr(P^\dagger P) = 1,\\&  \tr(Q^\dagger Q)=1 ,\\&  \tr(P^\dagger Q)=0
\end{aligned}
\right\}
\, \text{(Dual-unitary)}.\end{aligned} \label{eq:conditions_P_Q}
\end{align}
Consider the singular value decomposition $ P = V_1 D_1 W_1^\dagger $, where $ V_1 $ and $ W_1 $ are unitary matrices and $ D_1 = \text{Diag} \lbrace \sigma_1, \sigma_2 \rbrace$. The orthonormality condition $ \tr(P^\dagger P) =1 $ implies $ \sigma_1^2+\sigma_2^2=1 $. From the first two relations in Eq.~ (\ref{eq:conditions_P_Q}), we get
\begin{align}
\begin{aligned}
 Q Q^\dagger&=  V_1 D_1^{'2} V_1^\dagger,\\
  Q^\dagger Q &=W_1 D_1^{'2} W_1^\dagger, \\
\end{aligned}
\end{align}
where $ D_1' = \text{Diag} \lbrace \sigma_2, \sigma_1\rbrace $. Therefore, $ Q $ can be written as $ Q = V_1 D_2 W_1^\dagger $ where $ D_2$ is a diagonal matrix denoted by $ D_2 = \text{Diag} \lbrace \sigma'_1, \sigma'_2\rbrace $ and can in general be complex. Therefore, a local unitary transformation on $U_1$ given by
\begin{align}
\label{eq:step_2}
\begin{aligned}
U_2 & =\rl{ \mathbb{I}_3 \otimes \left(\begin{array}{cc}
	1&0\\
	0& V_1^\dagger 
	\end{array} \right) } U_1 \rl{ \mathbb{I}_3 \otimes \left(\begin{array}{cc}
	1&0\\
	0& W_1 
	\end{array} \right) } \\	
& =\left(\begin{array}{ccc|ccc|ccc}
1 & 0 & 0 & 0 & 0 & 0 & 0 & 0 & 0 \\
0 & 0 & 0 & 0 & \sigma_1 & 0& 0 & \sigma'_1& 0 \\
0 & 0 & 0 & 0 & 0 & \sigma_2& 0 & 0 & \sigma'_2 \\
\hline
0 & 0 & 0 & 0 & * & * & 0 & * & * \\
0 & * & * & * & 0 & 0 & * & 0 & 0 \\
0 & * & * & * & 0 & 0 & * & 0 & 0 \\
\hline
0 & 0 & 0 & 0 & * & * & 0 & * & * \\
0 & * & * & * & 0 & 0 & * & 0 & 0 \\
0 & * & * & * & 0 & 0 & * & 0 & 0 \\
\end{array}\right) \\&=\left(\begin{array}{ccc|ccc|ccc}
1 & 0 & 0 & 0 & 0 & 0 & 0 & 0 & 0 \\
0 & 0 & 0 & 0 & \sigma_1 & 0& 0 & \sigma'_1& 0 \\
0 & 0 & 0 & 0 & 0 & \sigma_2& 0 & 0 & \sigma'_2 \\
\hline
0 & 0 & 0 & 0 & * & * & 0 & * & * \\
0 & 0 & * & * & 0 & 0 & * & 0 & 0 \\
0 & * & 0 & * & 0 & 0 & * & 0 & 0 \\
\hline
0 & 0 & 0 & 0 & * & * & 0 & * & * \\
0 & 0 & * & * & 0 & 0 & * & 0 & 0 \\
0 & * & 0 & * & 0 & 0 & * & 0 & 0 \\
\end{array}\right).
\end{aligned}
\end{align}
The final form above has two more elements set to $0$ as the $3\times 3$ blocks have to be orthonormal.

It follows from the last two conditions in Eq.~(\ref{eq:conditions_P_Q}) that the matrix given by 
\begin{align}
\pmb{\sigma} = \begin{aligned}
\left(\begin{array}{cc}
\sigma_1 & \sigma_1'\\
\sigma_2&\sigma_2'
\end{array} \right),
\end{aligned}
\end{align}
is unitary. A local unitary transformation of the form
\begin{align}
\label{eq:step_3}
\begin{aligned}
U_3 =  U_2 \rl{ \begin{pmatrix}
	1&0\\0& \pmb{\sigma}^\dagger 
	\end{pmatrix} \otimes \mathbb{I}_3} 
\end{aligned}
\end{align}
followed by considering the unitarity of $ U_3^{\Gamma}$ results in 
\begin{align}
\begin{aligned}
U_3=\left(\begin{array}{ccc|ccc|ccc}
1 & 0 & 0 & 0 & 0 & 0 & 0 & 0 & 0 \\
0 & 0 & 0 & 0 & 1 & 0 & 0 & 0 & 0 \\
0 & 0 & 0 & 0 & 0 & 0& 0 &0& 1 \\
\hline
0 & 0 & 0 & 0 & 0 & c_{11} & 0 & c_{12} & 0 \\
0 & 0 & * & 0& 0 & 0 & * & 0 & 0 \\
0 & * & 0 & * & 0 & 0 & 0 & 0 & 0 \\
\hline
0 & 0 & 0 & 0 & 0 & c_{21} & 0 & c_{22} & 0\\
0 & 0 & * & 0 & 0 & 0 & * & 0 & 0 \\
0 & * & 0& * & 0 & 0 & 0 & 0 & 0 \\
\end{array}\right).
\end{aligned}
\end{align}
At this stage, we notice that {\em all} the columns of $U_3$ are unentangled. Therefore we need to apply local transformations on the left
to reduce further.
Note that four potentially non-zero entries are labeled $c_{ij}$ in $ U_3 $ form a $2 \times 2$ unitary matrix. Using this unitary matrix
\begin{align}
\label{eq:step_4}
\begin{aligned}
C = \begin{pmatrix}
c_{11}&c_{12}\\c_{21}&c_{22}
\end{pmatrix},
\end{aligned}
\end{align}
we apply the following local unitary transformation
\begin{align}
\begin{aligned}
U_4 = \rl{ \begin{pmatrix}
	1&0\\0& C^\dagger 
	\end{pmatrix} \otimes \mathbb{I}_3} U_3.
\end{aligned}
\end{align}
The constraint of $ U_4  $ being 2-unitary gives
\begin{align}
\begin{aligned}
U_4=\left(\begin{array}{ccc|ccc|ccc}
1 & 0 & 0 & 0 & 0 & 0 & 0 & 0 & 0 \\
0 & 0 & 0 & 0 & 1 & 0 & 0 & 0 & 0 \\
0 & 0 & 0 & 0 & 0 & 0& 0 &0& 1 \\
\hline
0 & 0 & 0 & 0 & 0 &1 & 0 & 0 & 0 \\
0 & 0 & 0 & 0& 0 & 0 & \alpha_2 & 0 & 0 \\
0 & \alpha_1 & 0 & 0 & 0 & 0 & 0 & 0 & 0 \\
\hline
0 & 0 & 0 & 0 & 0 & 0 & 0 & 1 & 0\\
0 & 0 & \alpha_3 & 0 & 0 & 0 & 0 & 0 & 0 \\
0 & 0 & 0& \alpha_4 & 0 & 0 & 0 & 0 & 0 \\
\end{array}\right),
\end{aligned}
\end{align}
where $ \alpha_1,\alpha_2,\alpha_3 $ and $ \alpha_4 $ are of modulus 1.\\
In the final step, we perform the following local unitary transformation
\begin{align} \label{eq:phase_lu}
\begin{aligned}
P_9  & = (\Phi_1 \otimes  \Phi_2 ) U_4( \Phi_3 \otimes  \Phi_4 ) \\
 & =\left(\begin{array}{ccc|ccc|ccc}
1 & 0 & 0 & 0 & 0 & 0 & 0 & 0 & 0 \\
0 & 0 & 0 & 0 & 1 & 0 & 0 & 0 & 0 \\
0 & 0 & 0 & 0 & 0 & 0& 0 &0& 1 \\
\hline
0 & 0 & 0 & 0 & 0 &1 & 0 & 0 & 0 \\
0 & 0 & 0 & 0& 0 & 0 & 1 & 0 & 0 \\
0 & 1 & 0 & 0 & 0 & 0 & 0 & 0 & 0 \\
\hline
0 & 0 & 0 & 0 & 0 & 0 & 0 & 1 & 0\\
0 & 0 & 1 & 0 & 0 & 0 & 0 & 0 & 0 \\
0 & 0 & 0& 1& 0 & 0 & 0 & 0 & 0 \\
\end{array}\right),
\end{aligned}
\end{align}
where
\begin{align}
\begin{aligned}
\Phi_1 & = \text{Diag} \left\{1,{\left(\alpha _1^* \alpha_2^* \right)^{\frac{1}{3}}}, \left(\alpha _3^* \alpha_4^* \right)^{\frac{1}{3}}  \right\} , \\
\Phi_2 & = \text{Diag} \left\{1,{\left(\alpha _2^* \alpha_3^* \right)^{\frac{1}{3}}}, \left(\alpha _1^* \alpha_4^* \right)^{\frac{1}{3}}  \right\}, \\
\Phi_3 & = \text{Diag} \left\{1,{\left(\alpha_1 \alpha _3 \alpha_4^* \right)^{\frac{1}{3}}}, \left(\alpha_1 \alpha _3 \alpha_2^* \right)^{\frac{1}{3}}  \right\} , \\
\Phi_4 & = \text{Diag} \left\{1,{\left(\alpha_2 \alpha _4 \alpha_1^* \right)^{\frac{1}{3}}}, \left(\alpha_2 \alpha _4 \alpha_3^* \right)^{\frac{1}{3}}  \right\} , \\
\end{aligned}
\end{align}
are diagonal unitaries.  Here, we choose the principal value of the cube root $ z^{1/3} $ with argument $ \arg(z) \in [0,2\pi) $. 
Concerning this last step, it has already been shown that any enphasing of $P_9$ is LU-equivalent to it \cite{Adam_SLOCC_2020}. 

{ Therefore, we have shown that, any 2-unitary two-qutrit operator is LU-equivalent to $P_9$. Hence, there exists only one LU-equivalence class of 2-unitary gates in $d=3$, or equivalently, there is only one AME$(4,3)$ state up to LU-equivalence. }

\end{document}